\documentclass[12pt]{amsart}
\usepackage{amssymb}
\usepackage{color}
\usepackage{amsmath,epic,curves,amscd}
\usepackage[english]{babel}
\usepackage{graphicx}
\usepackage{comment}
\usepackage{appendix}
\usepackage{mathdots}
\usepackage[all]{xy}
\usepackage{mathtools}
\usepackage{lscape}
\usepackage{stmaryrd}
\usepackage{mathabx}
\newtheorem{claim}{}[section]
\newtheorem{theorem}[claim]{Theorem}

\newtheorem{lemma}[claim]{Lemma}
\newtheorem{proposition}[claim]{Proposition}

\theoremstyle{remark}

\theoremstyle{example}

\renewenvironment{proof}{\noindent{\it Proof. \hskip0pt}}
                      {$\square$\par\medskip}
\allowdisplaybreaks \textwidth 15.5 true cm \textheight 24.0 true cm
\usepackage{color}

\begin{document}
\baselineskip 6.0 truemm
\parindent 1.5 true pc

\newcommand\lan{\langle}
\newcommand\ran{\rangle}
\newcommand\tr{\operatorname{Tr}}
\newcommand\ot{\otimes}
\newcommand\ttt{{\text{\sf t}}}
\newcommand\rank{\ {\text{\rm rank of}}\ }
\newcommand\choi{{\rm C}}
\newcommand\choid{{\rm D}}
\newcommand\dual{\star}
\newcommand\flip{\star}
\newcommand\cp{{{\mathbb C}{\mathbb P}}}
\newcommand\ccp{{{\mathbb C}{\mathbb C}{\mathbb P}}}
\newcommand\pos{{\mathcal P}}
\newcommand\tcone{T}
\newcommand\mcone{K}
\newcommand\superpos{{{\mathbb S\mathbb P}}}
\newcommand\blockpos{{{\mathcal B\mathcal P}}}
\newcommand\jc{{\text{\rm JC}}}
\newcommand\dec{{\mathbb D}{\mathbb E}{\mathbb C}}
\newcommand\decmat{{\mathcal D}{\mathcal E}{\mathcal C}}
\newcommand\ppt{{\mathcal P}}
\newcommand\pptmap{{\mathbb P}{\mathbb P}{\mathbb T}}
\newcommand\join{\vee}
\newcommand\meet{\wedge}
\newcommand\alg{{\text{\rm alg}}}
\newcommand\pr\prime
\newcommand\id{{\text{\rm id}}}
\newcommand\calm{{\mathcal M}}
\newcommand\calb{{\mathcal B}}
\newcommand\cala{{\mathcal A}}
\newcommand\cale{{\mathcal E}}
\newcommand\calf{{\mathcal F}}
\newcommand\calp{{\mathcal P}}
\newcommand\calt{{\mathcal T}}
\newcommand\cals{{\mathcal S}}
\newcommand\call{{\mathcal L}}
\newcommand\calr{{\mathcal R}}
\newcommand\caln{{\mathcal N}}
\newcommand\calcb{{\mathbb{CB}}}
\newcommand\calcp{{\mathbb{CP}}}
\newcommand\calnl{{\mathcal N\mathcal L}}
\newcommand\ad{{\text{\rm Ad}}}
\newcommand\ro{{\boxtimes}}
\newcommand\algmm{{\rm alg}(\calm^\pr,\calm)}
\newcommand\calhs{{\mathcal{HS}}}
\newcommand\rk{{\rm rank}\,}
\newcommand\e{\varepsilon}
\newcommand\conv{{\text{\rm conv}\,}}
\newcommand\re{{\text{\rm Re}}\,}
\newcommand\xxxx{\medskip\par ===============================\par\medskip}
\newcommand\im{{\text{\rm Im}}}
\newcommand\spa{{\text{\rm span}}\,}
\newcommand\cii{\circle*{0.4}}
\newcommand\inte{{\text{\rm int}\,}}

\title{Bi-qutrit entangled edge states of positive partial transposes with largest ranks}

\author{Kyung Hoon Han and Seung-Hyeok Kye}
\address{Kyung Hoon Han, Department of Data Science, The University of Suwon, Gyeonggi-do 445-743, Korea}
\email{kyunghoon.han at gmail.com}
\address{Seung-Hyeok Kye, Department of Mathematics and Institute of Mathematics, Seoul National University, Seoul 151-742, Korea}
\email{kye at snu.ac.kr}


\thanks{partially supported by NRF-RS-2025-23323242, Korea.}

\begin{abstract}
Whenever $E$ is an eight dimensional subspace of the bi-qutrit
quantum system whose orthogonal complement is spanned by a vector of
Schmidt rank three, we show that there exist PPT entangled edge
states with the range space $E$ whose partial transposes are of rank
six, which is the largest possible rank. In this way, we
exhibit a huge family of bi-qutrit PPT entangled edge
states of type $(8,6)$. They make faces of the convex set of all PPT
states, and we find bi-qutrit PPT entangled edge states of
other types on the boundaries of such faces.
\end{abstract}
\maketitle

\section{Introduction}

The notion of entanglement is now considered as one of the key
resources in the current quantum information theory, and
distinguishing entanglement from separability is one of the main
research topics. It is trivial to exhibit all separable states just
by definition, but it is very difficult to determine if a given
state is separable or entangled. The notion of positive partial
transposes introduced by Woronowicz \cite{woronowicz} in the 1970's
gives rise to a powerful necessary condition for separability; every
separable state is of PPT \cite{{choi-ppt},{peres}}. It is easy to
check if a given state is of PPT or not, but it is not so easy to
exhibit entangled PPT states. In fact, the whole structure of the
convex set $\ppt$ of all  PPT states is far from being understood
even in the low dimensional cases like $2\ot 4$ and $3\ot 3$
systems. In the $2\ot 2$ and $2\ot 3$ cases, it is known that the
notion of PPT coincides with separability
\cite{{stormer},{woronowicz},{horo-1}}, and their convex structures
can be found in \cite{{choi_kye},{kye_ritsu}}. PPT states themselves
play important roles in current quantum information theory. See
\cite{{halder},{horo-distill},{huber}} for examples.

The notion of PPT entangled edge states \cite{lkhc} is a key notion
to understand the convex structures of the convex set $\ppt$. Recall
that a PPT state $\varrho$ is called a PPT edge state, or just an
edge state, if the smallest face of $\ppt$ containing $\varrho$ has
no separable state. Then every PPT state is the sum of
a separable state and an edge state. Especially, every
extreme point of $\ppt$ is a pure separable state or an edge state.
Because every face of $\ppt$ is determined by a pair of subspaces
given by the images of a PPT state and its partial transpose, it is
natural to classify PPT edge states by ranks of themselves and their
partial transposes. A PPT state $\varrho$ is called of type $(h,k)$
when $\rk\varrho=h$ and $\rk\varrho^\Gamma=k$, where
$\varrho^\Gamma$ denotes the partial transpose of $\varrho$.

Lower bounds and upper bounds for possible types of edge states have
been discussed in \cite{{krausCKL_2000},{hlvc}} and
\cite{kiem_kye_11}, respectively. See \cite{kiem-multi} for
multi-partite cases. In case of $3\ot 3$ system, the early examples
of PPT entanglement by Choi \cite{choi-ppt} and St\o rmer
\cite{stormer82} in the early 1980's turn out to be PPT edge states of
types $(4,4)$ and $(7,6)$, respectively. Possible types for $3\ot 3$
edge states are listed as follows;
$$
(4, 4),\quad (5,5),\quad (6,5),\quad (6,6),\quad (7,5),\quad
(7,6),\quad  (8,5),\quad (8,6),
$$
where we consider the cases $\rk\varrho\ge\rk\varrho^\Gamma$ by
symmetry \cite{{kiem_kye_11},{kye_osaka}}. See {\sc Figure 1}.
\begin{figure}
\begin{center}
\setlength{\unitlength}{.5 truecm}
\begin{picture}(8,8)
\put(0,1){\line(1,0){7}} \put(1,0){\line(0,1){7}}
\put(1,4){\line(1,0){6}} \put(4,1){\line(0,1){6}}
\put(1,7){\line(1,0){6}} \put(7,1){\line(0,1){6}}
\put(7,3){\line(-1,1){4}} \put(2,2){\cii} \put(3,3){\cii}
\put(4,3){\cii} \put(5,3){\cii} \put(6,3){\cii} \put(3,4){\cii}
\put(4,4){\cii} \put(5,4){\cii} \put(6,4){\cii} \put(3,5){\cii}
\put(3,6){\cii} \put(4,5){\cii} \put(4,6){\cii} \put(1.1,0.2){$3$}
\put(3.9,0.2){$6$} \put(6.9,0.2){$9$}
\end{picture}
\end{center}
\caption{This picture shows the lower bounds $\rk\varrho\ge 4$ and
$\rk\varrho^\Gamma\ge 4$, together with the upper bound
$\rk\varrho+\rk\varrho^\Gamma\le 14$ for ranks of  $3\ot 3$ PPT
entangled edge states and their partial transposes.}
\end{figure}
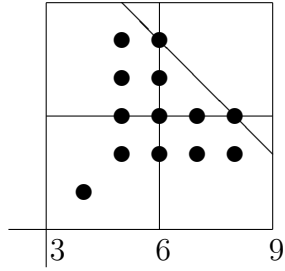
Bi-qutrit PPT edge states of type $(4,4)$ have been studied
extensively. In fact, it is known \cite{{LMS_2010},{chen},{sko},{SLM_2011},{chen_2012}} that every $3\ot 3$ PPT entangled state of rank
four is locally equivalent to that arising from unextendible product
basis \cite{{bdmsst},{dmsst}}. See also \cite{chen_2013} for general
$n\ot n$ cases. Nevertheless, examples of $3\ot 3$ PPT
edge states of other types are rather sporadic through the
literature \cite{{p-horo-range},{ha-kye-2},{clarisse},{ha-3},{agkl},{kye_osaka},{hansen_Myrh_2017}}.
See \cite{kye_ritsu} for a survey on the topics.

In order to understand the convex geometry of PPT states, we turn
our attention in this paper to PPT entangled edge
states of largest ranks. For a given subspace $E\subset \mathbb
C^m\ot\mathbb C^n$, the set
$$
{\mathcal F}_E:=\{\varrho\in \mathcal D_{m\otimes n}:
\im\varrho\subseteq E\}
$$
is a face of the convex set ${\mathcal D}_{m\otimes n}$ of all density
matrices in $M_m\ot M_n$, and every face of the convex set
${\mathcal D}_{m\otimes n}$ arises in this way, where $\im\varrho$
denotes the range space of $\varrho$. The (relative) interior of
${\mathcal F}_E$ with respect to the affine space
generated by ${\mathcal F}_{E}$ is given by
$$
\inte {\mathcal F}_E=\{\varrho\in {\mathcal D}_{m\otimes n}:
\im\varrho= E\}.
$$
It is known
\cite{{kiem_kye_11},{kye_ritsu},{choi_kiem_kye_2020}} that the
largest possible rank of an $n\ot n$ PPT edge state $\varrho$ is
$n^2-1$, and  the largest possible rank of $\varrho^\Gamma$ is
$n^2-(2n-3)$ for such states. Therefore, such a state should be
located on the interior of a maximal face of ${\mathcal D}_{n \ot
n}$.

The main purpose of this paper is to exhibit $3\ot 3$ PPT entangled
edge states of type $(8,6)$ with the range space $E$, whenever
$E^\perp$ is spanned by a vector of Schmidt rank three. Note that vectors of Schmidt rank three
cover generic cases among vectors in $\mathbb C^3\ot\mathbb C^3$.
After we suggest in the next section general forms of $3\ot 3$ PPT
states which fit our purpose, we find $3\ot 3$ PPT edge states
of type $(8,6)$ among them in Section 3, and consider the faces of the
convex set $\mathcal P$ determined by PPT states we found.
In Section 4, we find PPT edge states of other types
on the boundary of the face. In the last section, we discuss the cases
when $E^\perp$ is spanned by a vector of Schmidt rank one or two, together with several
questions.  Throughout this paper, we use the term \lq\lq state\rq\rq\ even for unnormalized states.

\section{Constructions}

In order to construct bi-qutrit PPT states of rank eight,
we begin with a kernel vector $|\zeta\ran\in\mathbb C^3\ot\mathbb C^3$
whose Schmidt rank is three.
We may assume without loss of generality that $|\zeta\ran=\sum_{i=0}^2 a_i|i\ran|i\ran$
with
\begin{equation}\label{a_i_cond}
\sum_{i=0}^2 a_i^2=1,\qquad a_i>0,\quad i=0,1,2,
\end{equation}
by Schmidt decomposition. Then it is natural to consider
$3\times 3$ rank two matrices
$$
A_{d,\alpha}:=
\left(\begin{matrix}
d_0 &\alpha_2 &\bar\alpha_1\\
\bar\alpha_2 &d_1 &\alpha_0\\
\alpha_1 &\bar\alpha_0 &d_2
\end{matrix}\right)
$$
with the kernel vector $(a_0,a_1,a_2)^\ttt\in\mathbb C^3$, which
will be the $3\times 3$ principal submatrices of $\varrho\in M_3\ot M_3$ with
the nonzero entries at $|00\ran$, $|11\ran$ and $|22\ran$. Considering the
partial transposes, it is also natural to look for PPT states among states
with the following form
\begin{equation}\label{varrho}
\varrho= \left(
\begin{array}{ccccccccccc}
d_0     &\cdot   &\cdot  &&\cdot  &\alpha_2     &\cdot   &&\cdot   &\cdot  &\bar\alpha_1     \\
\cdot   &p_2 &\cdot    &&\bar\beta_2    &\cdot   &\cdot &&\cdot &\cdot     &\cdot   \\
\cdot  &\cdot    &q_1 &&\cdot &\cdot  &\cdot    &&\beta_1    &\cdot &\cdot  \\
\\
\cdot  &\beta_2    &\cdot &&q_2 &\cdot  &\cdot    &&\cdot    &\cdot &\cdot  \\
\bar\alpha_2     &\cdot   &\cdot  &&\cdot  &d_1     &\cdot   &&\cdot   &\cdot  &\alpha_0     \\
\cdot   &\cdot &\cdot    &&\cdot    &\cdot   &p_0 &&\cdot &\bar\beta_0    &\cdot   \\
\\
\cdot   &\cdot &\bar\beta_1   &&\cdot    &\cdot   &\cdot &&p_1 &\cdot    &\cdot   \\
\cdot  &\cdot    &\cdot &&\cdot &\cdot  &\beta_0    &&\cdot    &q_0 &\cdot  \\
\alpha_1     &\cdot   &\cdot  &&\cdot  &\bar\alpha_0     &\cdot &&\cdot   &\cdot  &d_2
\end{array}
\right)\in M_3\ot M_3.
\end{equation}
Here, $\alpha=(\alpha_0,\alpha_1,\alpha_2)$ and $\beta=(\beta_0,\beta_1,\beta_2)$ are taken from $\mathbb C^3$,
and $d=(d_0,d_1,d_2)$, $p=(p_0,p_1,p_2)$ and $q=(q_0,q_1,q_2)\in \mathbb R^3_+$
with $\mathbb R_+=[0,\infty)$.
We denote
$$
B_{p_i,q_i,\beta_i}:=\left(\begin{matrix}
p_i &\bar\beta_i\\
\beta_i &q_i
\end{matrix}\right),\quad i=0,1,2,
$$
to express $2\times 2$ submatrices of $\varrho$.

We take $d=\alpha=(1,1,1)$ to recover PPT edge states by Choi
\cite{choi-ppt} and St\o rmer \cite{stormer82} of types $(4,4)$ and
$(7,6)$, respectively. We also note that examples of PPT edge states
of various types in \cite{{ha-kye-2},{kye_osaka}} are in the form
(\ref{varrho}). Furthermore, Werner states \cite{Werner-1989} and
isotropic states \cite{terhal-sghmidt} are in this form, together
with their variations in \cite{han_kye_2025_b}. The above matrices
arise as the Choi matrices of various kinds of positive linear maps
on $3\times 3$ matrices. See
\cite{{choi75},{choi-lam},{tomiyama-83},{tom_85},{cho-kye-lee}} for
early such examples. We also note that they are local diagonal
orthogonal invariant matrices
\cite{{nechita_singh_2021},{singh_nechita2001}}. See also
\cite{{park_youn_2024},{nechita_park_2025},{gulati_nechita_park}} in
this direction.

The partial transpose of $\varrho$ is given by
$$
\varrho^\Gamma= \left(
\begin{array}{ccccccccccc}
d_0     &\cdot   &\cdot  &&\cdot  &\beta_2     &\cdot   &&\cdot   &\cdot  &\bar\beta_1     \\
\cdot   &p_2 &\cdot    &&\bar\alpha_2    &\cdot   &\cdot &&\cdot &\cdot     &\cdot   \\
\cdot  &\cdot    &q_1 &&\cdot &\cdot  &\cdot    &&\alpha_1    &\cdot &\cdot  \\
\\
\cdot  &\alpha_2    &\cdot &&q_2 &\cdot  &\cdot    &&\cdot    &\cdot &\cdot  \\
\bar\beta_2     &\cdot   &\cdot  &&\cdot  &d_1     &\cdot   &&\cdot   &\cdot  &\beta_0     \\
\cdot   &\cdot &\cdot    &&\cdot    &\cdot   &p_0 &&\cdot &\bar\alpha_0    &\cdot   \\
\\
\cdot   &\cdot &\bar\alpha_1   &&\cdot    &\cdot   &\cdot &&p_1 &\cdot    &\cdot   \\
\cdot  &\cdot    &\cdot &&\cdot &\cdot  &\alpha_0    &&\cdot    &q_0 &\cdot  \\
\beta_1     &\cdot   &\cdot  &&\cdot  &\bar\beta_0     &\cdot
&&\cdot   &\cdot  &d_2
\end{array}
\right)\in M_3\ot M_3.
$$
Then we note that
$\varrho$ is of PPT if and only if
$$
A_{d,\alpha}, A_{d,\beta}\ge 0,
\qquad
B_{p_i,q_i,\alpha_i}, B_{p_i,q_i,\beta_i}\ge 0,\quad i=0,1,2,
$$
where $A\ge 0$ means that $A$ is positive (semi-definite).

We proceed to look for $\alpha_i=-b_i+{\rm i}c_i$ and $d=(d_0,d_1,d_2)$ such that $A_{d,\alpha}$ is a $3\times 3$
positive matrix of rank two  with the kernel vector $a=(a_0,a_1,a_2)^\ttt$. Looking at the imaginary parts of the equation
\begin{equation}\label{equation}
\left(\begin{matrix}
d_0 & -b_2+{\rm i}c_2 & -b_1-{\rm i}c_1\\
-b_2-{\rm i}c_2 &d_1 &-b_0+{\rm i}c_0\\
-b_1+{\rm i}c_1 & -b_0-{\rm i}c_0 &d_2
\end{matrix}\right)
\left(\begin{matrix} a_0\\ a_1\\ a_2\end{matrix}\right)
=
\left(\begin{matrix} 0\\ 0\\ 0\end{matrix}\right),
\end{equation}
we see that $c_i=ra_i$ for $i=0,1,2$ with $r\in\mathbb R$. We also consider the real parts, to see that
\begin{equation}\label{d_i}
d_0=\frac 1{a_0}(b_2a_{1}+b_{1}a_{2}),\quad
d_1=\frac 1{a_1}(b_0a_{2}+b_{2}a_{0}),\quad
d_2=\frac 1{a_2}(b_1a_{0}+b_{0}a_{1}).
\end{equation}
Therefore, we see that $A_{d,\alpha}$ is of the form
$$
A_{a,b,r}=\left(\begin{matrix}
d_0 & -b_2+{\rm i}ra_2 & -b_1-{\rm i}ra_1\\
-b_2-{\rm i}ra_2 &d_1 &-b_0+{\rm i}ra_0\\
-b_1+{\rm i}ra_1 & -b_0-{\rm i}ra_0 &d_2
\end{matrix}\right),
$$
with $b=(b_0,b_1,b_2)\in\mathbb R^3$ and $r\in\mathbb R$, where $d_0$, $d_1$ and $d_2$ are given by (\ref{d_i}).

By a direct computation, we note that the determinants of all $2\times 2$ principal submatrices of $A_{a,b,r}$ are given by
positive scalar multiplications of
$\delta-r^2$,
where
\begin{equation}\label{delta}
\delta:=\
\frac{b_0b_1}{a_0a_1}+
\frac{b_1b_2}{a_1a_2}+
\frac{b_2b_0}{a_2a_0}.
\end{equation}
Note that $A_{a,b,r}$ is positive if and only if all the $2\times 2$ submatrices are positive, since $A_{a,b,r}$
is singular. Therefore, we conclude that
$A_{a,b,r}$ is positive if and only if
$$
r^2\le\delta,\qquad {\text{\rm and}}\qquad
d_i\ge 0,\quad i=0,1,2
$$
hold. When this is the case, $A_{a,b,r}$  is of rank two if and only if $r^2<\delta$.
Proof of the following elementary lemma will be omitted.

\begin{lemma}\label{lemma_condition}
For real numbers $x_0$, $x_1$ and $x_2$ satisfying $x_0x_1+x_1x_2+x_2x_0\ge 0$, the following are equivalent;
\begin{enumerate}
\item[{\rm (i)}]
$x_0+x_1\ge 0$, $x_1+x_2\ge 0$ and $x_2+x_0\ge 0$,
\item[{\rm (ii)}]
$x_0+x_1+x_2\ge 0$,
\item[{\rm (iii)}]
at least two of $x_0,x_1,x_2$ are nonnegative.
\end{enumerate}
\end{lemma}

Suppose that $A_{a,b,r}$ is positive. Then
$r^2\le\delta$ implies $\delta\ge 0$.
With $x_i=\frac{b_i}{a_i}$ for $i=0,1,2$, we see that Lemma \ref{lemma_condition} (i) holds by $d_i\ge 0$ for $i=0,1,2$.
Therefore, we see that $b=(b_0,b_1,b_2)\in\mathbb R^3$ satisfies
\begin{equation}\label{b_condition}
\frac{b_0b_1}{a_0a_1}+
\frac{b_1b_2}{a_1a_2}+
\frac{b_2b_0}{a_2a_0}\ge 0,\qquad
\frac {b_0}{a_0}+\frac{b_1}{a_1}+\frac{b_2}{a_2}\ge 0.
\end{equation}
Conversely, suppose that $b\in\mathbb R^3$ satisfies (\ref{b_condition}). Then we have
$\delta\ge 0$ and $d_i\ge 0$ for $i=0,1,2$ by (ii) $\Longrightarrow$ (i) of Lemma \ref{lemma_condition} again.
Therefore, we have the following:

\begin{proposition}\label{main_sub}
A $3\times 3$ matrix is positive with a kernel vector $a=(a_0,a_1,a_2)$ satisfying {\rm (\ref{a_i_cond})} if and only if it is of the form
$A_{a,b,r}$ with $b\in\mathbb R^3$ and $r\in\mathbb R$ satisfying {\rm(\ref{b_condition})} and  $r^2\le \delta$, where $\delta$ is given by {\rm (\ref{delta})}.
When this is the case, we have the following;
\begin{enumerate}
\item[{\rm (i)}]
$A_{a,b,r}$ is of rank two if and only if $r^2<\delta$.
In this case, all the $2\times 2$ principal submatrices are positive definite.
\item[{\rm (ii)}]
$A_{a,b,r}$ is of rank one if and only if $r^2=\delta$ and $b\neq 0$,
\item[{\rm (iii)}]
$A_{a,b,r}$ is of rank zero if and only if $b=0$.
\end{enumerate}
\end{proposition}

As for the claims (i) and (ii) in Proposition \ref{main_sub}, we note that a $3\times 3$ positive singular matrix
is of rank one if and only if $2\times 2$ principal submatrices are of rank one.
We also note that a positive matrix $A_{a,b,r}$ is zero if and only if $d_i=0$ for $i=0,1,2$
if and only if $b_i=0$ for every $i=0,1,2$.
By the identity
$$
\frac{b_0b_1}{a_0a_1}+
\frac{b_1b_2}{a_1a_2}+
\frac{b_2b_0}{a_2a_0}=
\left({b_0 \over 2a_0} + {b_1 \over 2a_1} + {b_2 \over a_2}\right)^2 - \left({b_0 \over 2a_0} - {b_1 \over 2a_1}\right)^2 - \left(b_2 \over a_2\right)^2,
$$
we see that the relations for $b=(b_0,b_1,b_2)\in\mathbb R^3$ in (\ref{b_condition}) represent
one half of an elliptic cone, for each fixed $a=(a_0,a_1,a_2)$
satisfying (\ref{a_i_cond}). We have $\delta>0$ whenever $b$ is chosen in the interior of this elliptic cone.
If we take $b$ on the boundary of the elliptic cone, then we have $\delta=0$, which implies $r=0$. When $b$ is the origin, $A_{a,b,r}$ is of rank zero; otherwise it is of rank one.

As an example of our construction, we consider the case $a=b=\frac
1{\sqrt 3}(1,1,1)$. Then we have $\delta=3$, and
$d=\frac 2{\sqrt 3}(1,1,1)$. We also take $r=-\tan\theta$ with
$|\theta|\le \frac \pi 3$ so that $r^2 \le \delta$ and $\alpha_i=-\frac{1}{\sqrt{3}\cos \theta} e^{{\rm i}\theta}$.
Multiplying the constant $\sqrt 3\cos\theta$ to $A_{d,\alpha}$, we have
$$
\alpha_i=-e^{{\rm i}\theta},\qquad d_i=2\cos\theta,\qquad r=-\sqrt{3}\sin \theta, \qquad \delta = 9\cos^2\theta, \qquad i=0,1,2.
$$
By taking $\beta=(0,0,0)$, $p=(\frac 1\lambda,\frac 1\lambda,\frac 1\lambda)$, $q=(\lambda,\lambda,\lambda)$
with $\lambda>0$ in (\ref{varrho}) and $0<|\theta|<\frac \pi 3$, we recover the examples of PPT edge states
\begin{equation}\label{exam_KO}
\varrho_{\rm ko}(\theta,\lambda):= \left(
\begin{array}{ccccccccccc}
2\cos\theta     &\cdot   &\cdot       &\cdot  &-e^{{\rm i}\theta}     &\cdot       &\cdot   &\cdot  &-e^{-{\rm i}\theta}    \\
\cdot   &\frac 1\lambda &\cdot           &\cdot    &\cdot   &\cdot             &\cdot &\cdot     &\cdot   \\
\cdot  &\cdot    &\lambda         &\cdot &\cdot  &\cdot                 &\cdot    &\cdot &\cdot  \\
\cdot  &\cdot    &\cdot &\lambda&\cdot  &\cdot    &\cdot    &\cdot &\cdot  \\
-e^{-{\rm i}\theta}     &\cdot   &\cdot  &\cdot  &2\cos\theta     &\cdot   &\cdot   &\cdot  &-e^{{\rm i}\theta}     \\
\cdot   &\cdot &\cdot    &\cdot    &\cdot   &\frac 1\lambda &\cdot &\cdot    &\cdot   \\
\cdot   &\cdot &\cdot    &\cdot    &\cdot   &\cdot &\frac 1\lambda &\cdot    &\cdot   \\
\cdot  &\cdot    &\cdot &\cdot &\cdot  &\cdot    &\cdot    &\lambda &\cdot  \\
-e^{{\rm i}\theta}     &\cdot   &\cdot  &\cdot  &-e^{-{\rm i}\theta}     &\cdot  &\cdot   &\cdot  &2\cos\theta
\end{array}
\right)
\end{equation}
of type $(8,6)$ presented by the second author and Osaka in \cite{kye_osaka}.
When $\theta=\pm \frac \pi 3$, that is, $r=\mp \sqrt{\delta}$, each state $\varrho_{\rm ko}(\pm \frac \pi 3,\lambda)$
is locally equivalent to the example \cite{stormer82} of St\o rmer given by
\begin{equation}\label{exam_st}
\varrho_{\rm st}(\lambda)
:=\left(
\begin{array}{ccccccccccc}
1     &\cdot   &\cdot  &\cdot  &1     &\cdot   &\cdot   &\cdot  &1     \\
\cdot   &\frac 1 \lambda &\cdot    &\cdot    &\cdot   &\cdot &\cdot &\cdot     &\cdot   \\
\cdot  &\cdot    &\lambda &\cdot &\cdot  &\cdot    &\cdot    &\cdot &\cdot  \\
\cdot  &\cdot    &\cdot &\lambda &\cdot  &\cdot    &\cdot    &\cdot &\cdot  \\
1     &\cdot   &\cdot  &\cdot  &1     &\cdot   &\cdot   &\cdot  &1    \\
\cdot   &\cdot &\cdot    &\cdot    &\cdot   &\frac 1\lambda &\cdot &\cdot    &\cdot   \\
\cdot   &\cdot &\cdot    &\cdot    &\cdot   &\cdot &\frac 1\lambda &\cdot    &\cdot   \\
\cdot  &\cdot    &\cdot &\cdot &\cdot  &\cdot    &\cdot    &\lambda &\cdot  \\
1     &\cdot   &\cdot  &\cdot  &1     &\cdot   &\cdot   &\cdot &1
\end{array}
\right).
\end{equation}
Indeed, we take the $3\times 3$ diagonal unitary $u$ with diagonal entries
$1$, $e^{\frac 23\pi{\rm i}}$, $e^{-\frac 23\pi{\rm i}}$, to get
$$
(I_3\ot u^*)\varrho_{\rm ko}(\textstyle\frac \pi 3,\lambda)(I_3\ot u)=\varrho_{\rm st}(\lambda) \quad \text{and} \quad
(I_3\ot u)\varrho_{\rm ko}(-\textstyle\frac \pi 3,\lambda)(I_3\ot u^*)=\varrho_{\rm st}(\lambda).
$$

\section{Main results}

From now on, we take $a=(a_0,a_1,a_2)$ satisfying (\ref{a_i_cond}), and  $b=(b_0,b_1,b_2)\in\mathbb R^3$ so that
strict inequalities hold in (\ref{b_condition}). Take $\delta$ by (\ref{delta}), then we have $\delta>0$.
We also take $r\in\mathbb R$ with $r^2<\delta$, and put
$\alpha_i=-b_i+{\rm i}ra_i$ and $d_i$ by (\ref{d_i}), so that $A_{d,\alpha}$ is positive and of rank two. Finally, we take
$p_i,q_i\in \mathbb R_+$ and $\beta_i\in\mathbb C$ for $i=0,1,2$ so
that $A_{d,\beta}$ is positive definite and
$$
p_iq_i=|\alpha_i|^2>|\beta_i|^2,\qquad i=0,1,2
$$
holds. Note that $p_i$ and $q_i$ depend on the choice of $r$ as well as $a,b\in\mathbb R^3$.
Then, the state $\varrho$ given by (\ref{varrho}) is a PPT state
with
$$
\begin{aligned}
\ker \varrho &= {\rm span}\{a_0|00\ran + a_1|11\ran + a_2|22\ran\}, \\
\ker \varrho^\Gamma &= {\rm span} \{ q_2 |01\ran - \alpha_2|10\ran, q_0|12\ran - \alpha_0 |21\ran, q_1|20\ran - \alpha_1 |02\ran \}.
\end{aligned}
$$
We recall that the PPT state $\varrho$ is an entangled edge state if and only if
there exists no $|\xi\ran\ot|\eta\ran$ in the range of $\varrho$ such that
$|\bar\xi\ran\ot|\eta\ran$ is in the range of $\varrho^\Gamma$. Therefore,
the state $\varrho$ is an edge state of type $(8,6)$ when the following system
\begin{equation}\label{system}
\begin{aligned}
&a_0x_0y_0+a_1x_1y_1+a_2x_2y_2=0,\\
&\qquad\quad q_2\bar x_0y_1=\alpha_2\bar x_1y_0,\\
&\qquad\quad q_0\bar x_1y_2=\alpha_0\bar x_2y_1,\\
&\qquad\quad q_1\bar x_2y_0=\alpha_1\bar x_0y_2
\end{aligned}
\end{equation}
of equations has no nontrivial solution at all.

We proceed to show that the system (\ref{system}) of equations has a nontrivial solution if and only if $r=0$.
Suppose that $r=0$. By the second condition of (\ref{b_condition}), we have $b_i\ge 0$ for some $i=0,1,2$.
We assume that $b_0\ge 0$ without loss of generality.
Then, we have
$$
\alpha_0 = -b_0+ {\rm i}ra_0  = -b_0 \le 0.
$$
Set $x_0=y_0=0$.
Then, the system (\ref{system}) becomes
\begin{equation}\label{xxxxxxx}
\begin{aligned}
a_1 x_1 y_1 + a_2 x_2 y_2 &= 0, \\
\alpha_0 \bar x_2 y_1 - q_0 \bar x_1 y_2 &=0,
\end{aligned}
\end{equation}
which has a nontrivial solution $(y_1,y_2)$ if and only if
\begin{equation}\label{yyyyyyy}
\begin{vmatrix}
a_1 x_1 & a_2 x_2 \\
\alpha_0 \bar x_2 & -q_0 \bar x_1
\end{vmatrix}
= -a_1 q_0 |x_1|^2 - a_2 \alpha_0 |x_2|^2=0,
\end{equation}
for some $(x_1,x_2)$.
Since $\alpha_0 \le 0$, we can choose such nonzero $(x_1,x_2)$.

For the converse, we assume that $r \ne 0$.
Since
$$
{\rm Im} (\alpha_0 \alpha_1 \alpha_2) = (b_0b_1a_2 + b_0 a_1 b_2 + a_0 b_1 b_2)r - a_0a_1a_2r^3
= a_0a_1a_2 (\delta -r^2) r
$$
and $0<r^2<\delta$, it follows that $\alpha_0\alpha_1\alpha_2 \notin \mathbb R$.
Multiplying all the equations in (\ref{system}) except the first, we obtain
\begin{equation}\label{eq_1}
q_0 q_1 q_2 \bar x_0 \bar x_1 \bar x_2 y_0 y_1 y_2 = \alpha_0 \alpha_1 \alpha_2 \bar x_0 \bar x_1 \bar x_2 y_0 y_1 y_2.
\end{equation}
Since $\alpha_0\alpha_1\alpha_2 \notin \mathbb R$, we have
$\bar x_0 \bar x_1 \bar x_2 y_0 y_1 y_2 = 0$.
That is, at least one of $x_i$ or $y_i$ is zero.
Without loss of generality, we suppose $x_0=0$.
Since $\alpha_2 \ne 0$ and $q_1>0$, the second and fourth equations in (\ref{system}) imply
$\bar x_1 y_0 = 0$ and $\bar x_2 y_0 = 0$, respectively.
If $(x_0,x_1,x_2) \ne (0,0,0)$, then we must have $y_0=0$.
The system (\ref{system}) then reduces to (\ref{xxxxxxx}) again.
In this case, the determinant (\ref{yyyyyyy}) is never zero,
because $r \ne 0$ implies $\alpha_0 \notin \mathbb R$.
Therefore, we conclude that $r\neq 0$ implies that the system (\ref{system}) admits only the trivial solution.
Now, we summarize our discussions as follows:

\begin{theorem}\label{main}
Take real numbers $a_i,b_i$ with $i=0,1,2$ and $r$ satisfying
$$
a_0^2+a_1^2+a_2^2=1,\quad a_i>0,\quad
\frac {b_0}{a_0}+\frac{b_1}{a_1}+\frac{b_2}{a_2}\ge 0,\quad
0<r^2<\delta,
$$
with $\delta$ given by {\rm (\ref{delta})}, and put $\alpha_i\in\mathbb C$, $d_i\in\mathbb R$ by
$\alpha_i=-b_i+{\rm i}ra_i$ and {\rm (\ref{d_i})} for $i=0,1,2$, respectively.
We also take $p_i,q_i\in \mathbb R_+$ and $\beta_i\in\mathbb C$ for $i=0,1,2$ so that $A_{d,\beta}$ is positive definite and
$$
p_iq_i=|\alpha_i|^2>|\beta_i|^2,\qquad i=0,1,2
$$
holds. Then the state $\varrho$ given by {\rm (\ref{varrho})} is a PPT entangled edge state of type $(8,6)$.
\end{theorem}

If we take $r=0$ in Theorem \ref{main} instead of $0<r^2<\delta$, then we have the state
$$
\varrho=
\left(
\begin{array}{ccccccccccc}
b_1{a_2 \over a_0} + b_2 {a_1 \over a_0}  &\cdot   &\cdot  &&\cdot  &-b_2     &\cdot   &&\cdot   &\cdot  &-b_1     \\
\cdot   &p_2 &\cdot    &&\bar \beta_2    &\cdot   &\cdot &&\cdot &\cdot     &\cdot   \\
\cdot  &\cdot    &q_1 &&\cdot &\cdot  &\cdot    &&\beta_1    &\cdot &\cdot  \\
\\
\cdot  &\beta_2    &\cdot &&q_2 &\cdot  &\cdot    &&\cdot    &\cdot &\cdot  \\
-b_2     &\cdot   &\cdot  &&\cdot  & b_2{a_0 \over a_1} + b_0 {a_2 \over a_1}     &\cdot   &&\cdot   &\cdot  &-b_0     \\
\cdot   &\cdot &\cdot    &&\cdot    &\cdot   &p_0 &&\cdot &\bar \beta_0    &\cdot   \\
\\
\cdot   &\cdot &\bar \beta_1   &&\cdot    &\cdot   &\cdot &&p_1 &\cdot    &\cdot   \\
\cdot  &\cdot    &\cdot &&\cdot &\cdot  &\beta_0    &&\cdot    &q_0 &\cdot  \\
-b_1     &\cdot   &\cdot  &&\cdot  &-b_0     &\cdot &&\cdot   &\cdot  &b_0{a_1 \over a_2} + b_1 {a_0 \over a_2}
\end{array}
\right).
$$
This state is decomposed into the sum of the following three states
$$
\begin{aligned}
&b_2 {a_1 \over a_0}|00\ran\lan00| + p_2|01\ran\lan01| + q_2|10\ran\lan10| + b_2{a_0 \over a_1}|11\ran\lan11|\\
&\qquad\qquad\qquad - b_2|00\ran\lan11| +\bar \beta_2|01\ran\lan10| + \beta_2|10\ran\lan01| - b_2|11\ran\lan00|, \\
&b_0 {a_2 \over a_1}|11\ran\lan11| + p_0|12\ran\lan12| + q_0|21\ran\lan21| + b_0{a_1 \over a_2}|22\ran\lan22|\\
&\qquad\qquad\qquad - b_0|11\ran\lan22| +\bar \beta_0|12\ran\lan21| + \beta_0|21\ran\lan12| - b_0|22\ran\lan11|, \\
&b_1 {a_0 \over a_2}|22\ran\lan22| + p_1|20\ran\lan20| + q_1|02\ran\lan02| + b_1{a_2 \over a_0}|00\ran\lan00|\\
&\qquad\qquad\qquad - b_1|22\ran\lan00| +\bar \beta_1|20\ran\lan02| + \beta_1|02\ran\lan20| - b_1|00\ran\lan22|,
\end{aligned}
$$
and the nonzero part of each summand makes a $2 \otimes 2$ PPT state
$$
\begin{pmatrix}
b_k {a_j \over a_i}&\cdot&\cdot&-b_k \\
\cdot&p_k&\bar \beta_k&\cdot \\
\cdot&\beta_k&q_k&\cdot \\
-b_k&\cdot&\cdot&b_k {a_i \over a_j}
\end{pmatrix}
$$
for $(i,j,k)=(0,1,2), (1,2,0), (2,0,1)$. Therefore, we see that $\varrho$ is separable.

If the state $\varrho$ in Theorem \ref{main} satisfies $r^2=\delta$ instead of $0<r^2<\delta$, then it is still an edge state of type (7,6),
provided it additionally satisfies the condition $q_0q_1q_2\neq\alpha_0\alpha_1\alpha_2$,
by exactly the same argument used around (\ref{eq_1}) in the proof of Theorem \ref{main}.
If $q_0q_1q_2=\alpha_0\alpha_1\alpha_2$, however, it may be separable.
For example, the states $\varrho_{\rm ko}(\pm \frac \pi 3,1)$ in (\ref{exam_KO}) satisfy $r^2=\delta$ and $q_0q_1q_2=\alpha_0\alpha_1\alpha_2$.
Since it is locally equivalent to the state $\varrho_{\rm st}(1)$ in (\ref{exam_st}), it suffices to show that $\varrho_{\rm st}(1)$ is separable.
To see this, we put
$$
|\zeta_{\alpha,\beta}\ran=(|0\ran + \alpha |1\ran + \beta|2\ran) \ot (|0\ran + \bar \alpha |1\ran + \bar \beta|2\ran),
$$
and average $|\zeta_{\alpha,\beta}\ran\lan\zeta_{\alpha,\beta}|$ over $\alpha,\beta=\pm 1,\pm{\rm i}$ to obtain $\varrho_{\rm st}(1)$.

In the remainder of this section, we consider the faces of the convex set $\ppt$ of all PPT states,
whose interiors consist of PPT edge states of type $(8,6)$.
Suppose that there exists a PPT state $\varrho$ with $\im\varrho=D$ and $\im\varrho^\Gamma=E$. For such a pair $(D,E)$ of subspaces,
we consider the convex set $\tau(D,E)$
of all  PPT states $\varrho$ in $M_3\ot M_3$ satisfying $\im\varrho\subset D$ and $\im\varrho^\Gamma\subset E$.
Then $\tau(D,E)$ is a face of the convex set $\ppt$, and every face of $\ppt$ arises in this way.
See \cite{kye_ritsu}.

We fix an eight dimensional subspace orthogonal to a vector $\sum_{i=0}^2a_i|ii\ran$ satisfying
(\ref{a_i_cond}).
Now, we begin with the parameter $\Lambda=(b,r,p)$ satisfying the conditions
$$
\frac {b_0}{a_0}+\frac{b_1}{a_1}+\frac{b_2}{a_2}\ge 0,\qquad
0<r^2<\delta:=\frac{b_0b_1}{a_0a_1}+
\frac{b_1b_2}{a_1a_2}+
\frac{b_2b_0}{a_2a_0},
$$
in Theorem \ref{main}, by which
$\alpha\in\mathbb C^3$ and $q\in\mathbb R^3_+$ are also determined as
\begin{equation}\label{alphaq}
\alpha_i=-b_i+{\rm i}ra_i, \qquad p_iq_i=|\alpha_i|^2.
\end{equation}
Then $\Lambda$ determines the pair $(D,E)$ of subspaces of $\mathbb C^3\ot\mathbb C^3$ by
$$
\begin{aligned}
D^\perp&=\spa\{a_0|00\ran+a_1|11\ran+a_2|22\ran\},\\
E^\perp&=\spa\{
q_2|01\ran-\alpha_2|10\ran,\
q_0|12\ran-\alpha_0|21\ran,\
q_1|20\ran-\alpha_1|02\ran\},
\end{aligned}
$$
with $\dim D=8$ and $\dim E=6$, and we obtain a PPT edge state $\varrho$ satisfying $\im\varrho=D$ and $\im\varrho^\Gamma=E$
by taking $d$ as (\ref{d_i}) and $\beta=0$, according to Theorem \ref{main}.
We denote by $F_{\Lambda}$ the face determined by the pair $(D,E)$ of subspaces.
We also denote by $\varrho_{\Lambda,\beta}$ the self-adjoint matrix in (\ref{varrho})
 where $\alpha\in\mathbb C^3$ and $q,d \in \mathbb R^3_+$ are determined by (\ref{alphaq}) and (\ref{d_i}).
Then $\varrho_{\Lambda,0}$ is an interior point of the face $F_\Lambda$.
The following proposition shows that the parameter $\Lambda$ determines the face $F_\Lambda$ up to positive scalar
multiplications.

\begin{proposition}\label{parame}
For given two parameters $\Lambda = (b,r,p)$ and $\Lambda' = (b',r',p')$, the following are equivalent:
\begin{enumerate}
\item[{\rm (i)}]
$F_\Lambda = F_{\Lambda'}$,
\item[{\rm (ii)}]
${r' \over r} \Lambda=\Lambda'$,
\item[{\rm (iii)}]
${r' \over r} \varrho_{\Lambda,0} = \varrho_{\Lambda',0}$.
\end{enumerate}
\end{proposition}

\begin{proof}
Suppose that $F_\Lambda = F_{\Lambda'}$. Then we have $\tau(D,E) = F_\Lambda = F_{\Lambda'} = \tau(D,E')$, where
$$
\begin{aligned}
E^\perp &= {\rm span} \{ q_2 |01\ran - \alpha_2 |10\ran, q_0 |12\ran - \alpha_0 |21\ran, q_1 |20\ran - \alpha_1 |02\ran \}, \\
E'^\perp &= {\rm span} \{ q_2' |01\ran - \alpha_2' |10\ran, q_0' |12\ran - \alpha_0' |21\ran, q_1' |20\ran - \alpha_1' |02\ran \}.
\end{aligned}
$$
Then there exist $\lambda_i>0$, $i=0,1,2$, such that
$$
\begin{aligned}
\lambda_2(q_2 |01\ran - \alpha_2 |10\ran) &= q_2' |01\ran - \alpha_2' |10\ran, \\
\lambda_0(q_0 |12\ran - \alpha_0 |21\ran) &= q_0' |12\ran - \alpha_0' |21\ran, \\
\lambda_1(q_1 |20\ran - \alpha_1 |02\ran)  &= q_1' |20\ran - \alpha_1' |02\ran.
\end{aligned}
$$
Since $\alpha_i = -b_i + {\rm i}r a_i$ and $\alpha_i' = -b_i' + {\rm i}r' a_i$, it follows that $\lambda_i = r'/r$ for each $i=0,1,2$.
This implies ${r' \over r} b_i = b_i'$ and $p_i' = {|\alpha_i'|^2 \over q_i'} = \frac{\lambda_i^2|\alpha_i|^2}{\lambda_i q_i} = {r' \over r} p_i$,
which proves (i) $\Longrightarrow$ (ii).

For the direction (ii) $\Longrightarrow$ (iii), we note
$$
\begin{aligned}
{r' \over r} d_i &= {(r'/r)(a_jb_k + a_kb_j) \over a_i} = {a_j b_k' + a_k b_j' \over a_i} = d_i', \\
{r' \over r} \alpha_i &= {r' \over r}(-b_i + {\rm i} r a_i) = -b_i' + {\rm i} r' a_i = \alpha_i', \\
{r' \over r} q_i &= {(r'/r)^2 |\alpha_i|^2 \over (r'/r)p_i} = {|\alpha_i'|^2 \over p_i'} = q_i',
\end{aligned}
$$
for distinct $i,j,k = 0,1,2$, which implies ${r' \over r} \varrho_{\Lambda,0} = \varrho_{\Lambda',0}$.
The direction (iii) $\Longrightarrow$ (i) is clear, since two faces coincide whenever they share a common interior point.
\end{proof}

Whenever $E$ is an eight dimensional space whose orthogonal complement is spanned  by
a vector of Schmidt rank three, we obtain a six parameter family of faces of $\mathcal P$ whose interiors consist of PPT edge states of type $(8,6)$
on the maximal face ${\mathcal F}_E$ of the convex set ${\mathcal D}_{3\ot 3}$.
It should be noted that condition (iii) of Proposition \ref{parame} implies that $r$ and $r^\prime$ share the same sign. In fact,
two parameters $(b,r,p)$ and $(b,-r,p)$ give rise to different faces.

\section{PPT edge states of other types}

In this section, we fix $a=(a_0,a_1,a_2)$ satisfying (\ref{a_i_cond}) and a parameter $\Lambda=(b,r,p)$,
and look for PPT edge states of other types on the boundary of the face $F_\Lambda$.
A self-adjoint matrix $\varrho_{\Lambda,\beta}$ belongs to the face $F_\Lambda$ if and only if it is of PPT
if and only if $\beta \in\mathbb C^3$ satisfies
\begin{equation}\label{qqq}
A_{d,\beta}\ge 0,\qquad
B_{i,\beta_i}:=B_{p_i,q_i,\beta_i} \ge 0\quad (i=0,1,2).
\end{equation}
When this is the case, we note that $\varrho_{\Lambda,\beta}$ is a PPT edge state automatically,
because the face $F_\Lambda$ itself contains no separable state.
For a PPT state $\varrho_{\Lambda,\beta}\in F_\Lambda$, we have
$$
\begin{aligned}
\rk\varrho_{\Lambda,\beta}&=2+\rk B_{0,\beta_0}+\rk B_{1,\beta_1}+\rk B_{2,\beta_2},\\
\rk\varrho_{\Lambda,\beta}^\Gamma&=3+\rk A_{d,\beta}.
\end{aligned}
$$
By choosing suitable $\beta\in\mathbb C^3$, we proceed to obtain PPT edge states of type $(h,k)$ with $h,k\ge 5$.

For the variable $z=(z_0,z_1,z_2)\in\mathbb C^3$, we define functions
$$
\begin{aligned}
F(z)&=\det A_{d,z}=d_0d_1d_2+2\re(z_0z_1z_2)-(d_0|z_0|^2+d_1|z_1|^2+d_2|z_2|^2),\\
G_i(z_i)&=\det B_{i,z_i}=p_iq_i-|z_i|^2=|\alpha_i|^2-|z_i|^2,\quad i=0,1,2.
\end{aligned}
$$
We consider the curve
$$
\begin{aligned}
0 &= F(\alpha_0,\alpha_1,z_2) \\
& = d_0d_1d_2 + 2{\rm Re}(\alpha_0\alpha_1)x + 2{\rm Re}({\rm i}\alpha_0\alpha_1)y
-(d_0|\alpha_0|^2 + d_1|\alpha_1|^2 + d_2(x^2+y^2)),
\end{aligned}
$$
which represents a circle on the complex plane through $\alpha_2$.
The second coordinate of its center is given by
$$
{{\rm Re}({\rm i}\alpha_0\alpha_1 ) \over d_2} = {{\rm Re}\left({\rm i}(-b_0+{\rm i}ra_0)(-b_1+{\rm i}ra_1) \right) \over d_2}
= {a_0b_1r + a_1b_0r \over d_2}
= a_2 r = {\rm Im} \alpha_2.
$$
Thus, the tangent line at $\alpha_2$ is parallel to the imaginary axis. On the other hand, the curve $G_2(z_2)=0$
represents the circle through $\alpha_2$ centered at the origin.
Since the imaginary part of $\alpha_2$ is nonzero by the assumption $r\neq 0$,
the tangent line of the circle at $\alpha_2$ is not parallel to the imaginary axis.
 See {\sc Figure 2}.
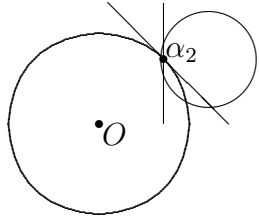
\begin{figure}
\begin{center}
\setlength{\unitlength}{.4 truecm}
\begin{picture}(9,7)
\qbezier(0,3)(0.2,0.2)(3,0)
\qbezier(3,0)(5.8,0.2)(6,3)
\qbezier(6,3)(5.8,5.8)(3,6)
\qbezier(3,6)(0.2,5.8)(0,3)
\put(6.63,5.13){\circle{3}}
\put(5.13,5.13){\circle*{0.3}}
\put(5.2,5.2){$\alpha_2$}
\put(5.13,7){\line(0,-1){4}}
\put(3.29,7){\line(1,-1){4}}
\put(3,3){\circle*{0.3}}
\put(3.1,2.3){$O$}
\end{picture}
\end{center}
\caption{If two circles share a point $\alpha_2$ at which they have different tangent lines, then corresponding open discs
have nonempty intersection.}
\end{figure}
Therefore, the open discs by these two circles have nonempty intersection, and
we can choose $\beta_2$ close to $\alpha_2$ satisfying
$$
d_0d_1 - |\beta_2|^2 > 0, \quad F(\alpha_0,\alpha_1,\beta_2)>0, \quad G_2(\beta_2)=0.
$$
Note that the first inequality is possible, because
all the $2\times 2$ principal submatrices of $A_{d,\alpha}$ are positive definite, which implies $d_0d_1 - |\alpha_2|^2 >0$.
Putting $\beta=(\alpha_0,\alpha_1,\beta_2)$,
the conditions in (\ref{qqq}) are satisfied, and we have
$$
\begin{aligned}
\rk\varrho_{\Lambda,\beta}&=2+\rk B_{0,\alpha_0}+\rk B_{1,\alpha_1}+\rk B_{2,\beta_2} = 2+1+1+1=5, \\
\rk\varrho_{\Lambda,\beta}^\Gamma&=3+\rk A_{d,(\alpha_0,\alpha_1,\beta_2)} = 3+ 3 = 6.
\end{aligned}
$$
We can also take $\beta_2'$ close to $\alpha_2$ satisfying
$$
d_0d_1 - |\beta^\prime_2|^2 > 0, \quad F(\alpha_0,\alpha_1,\beta_2')=0, \quad G_2(\beta_2')>0,
$$
and put $\beta=(\alpha_0,\alpha_1,\beta_2')$ to get
$$
\begin{aligned}
\rk\varrho_{\Lambda,\beta}&=2+\rk B_{0,\alpha_0}+\rk B_{1,\alpha_1}+\rk B_{2,\beta_2'} = 2+1+1+2=6, \\
\rk\varrho_{\Lambda,\beta}^\Gamma&=3+\rk A_{d,(\alpha_0,\alpha_1,\beta_2')} = 3+ 2 = 5.
\end{aligned}
$$
In this way, we get PPT edge states of type $(5,6)$ and $(6,5)$ on the boundary of the face $F_\Lambda$.

Now, we consider the circles defined by $F(\alpha_0,z_1,\alpha_2)=0$ and $G_1(z_1)=0$, both of which pass through $z_1=\alpha_1$.
By the same arguments as above, the open discs by these two circles have nonempty intersection.
If $\beta_2'$ was chosen sufficiently close to $\alpha_2$, we may replace $F(\alpha_0,z_1,\alpha_2)=0$ by
$F(\alpha_0,z_1,\beta_2')=0$ while preserving the same properties.
Thus, we can take $\beta_1$ close to $\alpha_1$ satisfying
$$
d_0d_2 - |\beta_1|^2 >0, \quad  F(\alpha_0,\beta_1,\beta_2')>0, \quad G_1(\beta_1)=0.
$$
Then conditions in (\ref{qqq}) are satisfied with $\beta=(\alpha_0,\beta_1,\beta_2')$ and we get
$$
\begin{aligned}
\rk\varrho_{\Lambda,\beta}&=2+\rk B_{0,\alpha_0}+\rk B_{1,\beta_1}+\rk B_{2,\beta_2'} = 2+1+1+2=6, \\
\rk\varrho_{\Lambda,\beta}^\Gamma&=3+\rk A_{d,(\alpha_0,\beta_1,\beta_2')} = 3+ 3 = 6.
\end{aligned}
$$
On the other hand, we take $\beta_1'$ close to $\alpha_1$ satisfying
$$
d_0d_2 - |\beta_1'|^2 >0, \quad F(\alpha_0,\beta_1',\beta_2')=0, \quad G_1(\beta_1')>0,
$$
and put $\beta=(\alpha_0,\beta_1',\beta_2')$.
Then we have
$$
\begin{aligned}
\rk\varrho_{\Lambda,\beta}&=2+\rk B_{0,\alpha_0}+\rk B_{1,\beta_1'}+\rk B_{2,\beta_2'} = 2+1+2+2=7, \\
\rk\varrho_{\Lambda,\beta}^\Gamma&=3+\rk A_{d,(\alpha_0,\beta_1',\beta_2')} = 3+ 2 = 5.
\end{aligned}
$$
If both $\beta_1'$ and $\beta_2'$ are chosen sufficiently close to $\alpha_1$ and $\alpha_2$,
we can also consider the curves $F(z_0,\beta_1',\beta_2')=0$ and $G_0(z_0)=0$.
Then, there exist $\beta=(\beta_0,\beta_1',\beta_2')$ satisfying
$$
\begin{aligned}
\rk\varrho_{\Lambda,\beta}&=2+\rk B_{0,\beta_0}+\rk B_{1,\beta_1'}+\rk B_{2,\beta_2'} = 2+1+2+2=7, \\
\rk\varrho_{\Lambda,\beta}^\Gamma&=3+\rk A_{d,(\beta_0,\beta_1',\beta_2')} = 3+ 3 = 6.
\end{aligned}
$$
and $\beta=(\beta_0',\beta_1',\beta_2')$ satisfying
$$
\begin{aligned}
\rk\varrho_{\Lambda,\beta}&=2+\rk B_{0,\beta_0'}+\rk B_{1,\beta_1'}+\rk B_{2,\beta_2'} = 2+2+2+2=8, \\
\rk\varrho_{\Lambda,\beta}^\Gamma&=3+\rk A_{d,(\beta_0',\beta_1',\beta_2')} = 3+ 2 = 5.
\end{aligned}
$$

In order to get PPT edge states of type $(5,5)$ on the boundary of the face $F_\Lambda$,
we perturb the angular parts of $\alpha_i$ to get $\beta_i$ for $i=0,1,2$,
satisfying the relation $\re(\beta_0\beta_1\beta_2)=\re(\alpha_0\alpha_1\alpha_2)$.
Then we have $F(\beta)=F(\alpha)=0$ and $G_i(\beta_i)=G_i(\alpha_i)=0$, which imply
$$
\begin{aligned}
\rk\varrho_{\Lambda,\beta}&=2+\rk B_{0,\alpha_0}+\rk B_{1,\alpha_1}+\rk B_{2,\alpha_2} = 2+1+1+1=5, \\
\rk\varrho_{\Lambda,\beta}^\Gamma&=3+\rk A_{d,\alpha} = 3+ 2 = 5.
\end{aligned}
$$
Therefore, we get a family of PPT edge states of type $(5,5)$ on the boundary of the face $F_\Lambda$

It remains to find PPT edge states of type $(4,4)$.
If the face $F_\Lambda$ contains $\varrho_{\Lambda,\beta}$ of rank four,
then one of $B_{i,\beta_i}$, say $B_{0,\beta_0}$, must be zero.
This implies  $\alpha_0=0$ by the PPT condition, and we have $r=0$.
Therefore, we see that the face $F_\Lambda$ cannot contain a PPT edge state of type $(4,4)$ with form in (\ref{varrho}).
On the other hand, Choi's original example of PPT entanglement \cite{choi-ppt} is a PPT edge state of type $(4,4)$ with form in (\ref{varrho}).
It is natural to ask if the face $F_\Lambda$ contains any other PPT edge states of type $(4,4)$ on the boundary.

\section{Kernel vector of Schmidt rank two or one}

In this last section, we consider the case when $|\zeta\ran=\sum_{i=0}^2 a_i|i\ran|i\ran$ is of Schmidt rank two or one,
and discuss some problems. We begin with the case
of Schmidt rank two, say, $a_0,a_2>0$ and $a_1=0$.
We begin with the equation (\ref{equation}) to get
$$
c=(ra_0,0,ra_2),\qquad
b=(-sa_0,b_1,s a_2),\qquad
d=(\textstyle\frac{a_2b_1}{a_0}, d_1,\frac{a_0b_1}{a_2})
$$
for $r,s \in \mathbb R$.
We take $b_1=a_0a_2$ to get
$$
A_{d,\alpha}=
\left(\begin{matrix}
a_2^2 & -a_2(s-{\rm i}r) & -a_0a_2\\
-a_2(s+{\rm i}r) & d_1 & a_0(s+{\rm i}r)\\
-a_0a_2 & a_0(s-{\rm i}r) &a_0^2
\end{matrix}\right).
$$
Note that $A_{d,\alpha}$ is positive if and only if
\begin{equation}\label{con1}
r^2+s^2\le d_1.
\end{equation}
When we take $\beta=0$, the state $\varrho$ in (\ref{varrho}) is of PPT if and only if (\ref{con1}) and
\begin{equation}\label{con2}
p_0q_0\ge a_0^2(r^2+s^2),\qquad
p_1q_1\ge a_0^2a_2^2,\qquad
p_2q_2\ge a_2^2(r^2+s^2)
\end{equation}
are satisfied.
If we take $r,s\in \mathbb R$ so that the strict inequality holds in (\ref{con1}) and take $p,q\in\mathbb R^3_+$ so that
the identities hold in (\ref{con2}), then the PPT state $\varrho$ in (\ref{varrho}) is of type (8,6).

Now, we consider the equation (\ref{system}) with $a_1=0$. Then we have the identity (\ref{eq_1}).
We note that $\alpha_0\alpha_1\alpha_2=a_0^2a_2^2(r^2+s^2)$. If we take $q\in\mathbb R^3_+$ so that $q_0q_1q_2\neq a_0^2a_2^2(r^2+s^2)$
then we have either $(x_0,y_0)=(0,0)$, $(x_1,y_1)=(0,0)$ or $(x_2,y_2)=(0,0)$.
In the first and the third cases, the nontrivial solutions are of the form $x=(0,x_1,0)$ and $y=(0,y_1,0)$,
while in the second case they are of the form $x=(x_0,0,x_2)$ and $y=(y_0,0,y_2)$.
Since $\lan0| \lan1|x\ran|y\ran = 0$, such product vectors $x \otimes y$ cannot span the orthogonal complement of $|\zeta\ran$.
Thus, the states $\varrho$ in (\ref{varrho}) violate the range criterion \cite{p-horo-range}, and we conclude that they are PPT entangled states,
even though they are not PPT edge states.
The authors could not determine if there exists a PPT edge state of type $(8,6)$ whose kernel vector has Schmidt rank two.

Now, we turn our attention to the case when the kernel vector is of Schmidt rank one.
More generally, we consider the case of $n\ot n$ system. Recall that the rank of $\varrho^\Gamma$ is at most $n^2-2n+3$
when a PPT edge state $\varrho$ has corank one by \cite{kiem_kye_11}. See also \cite{choi_kiem_kye_2020}.

\begin{proposition}\label{rankoneedge}
Suppose that $D$ is a maximal subspace of $\mathbb C^n\ot\mathbb C^n$ whose orthogonal complement
is spanned by a vector of Schmidt rank one. Then there exists no $n\ot n$ PPT edge state $\varrho$
of type $(n^2-1,n^2-2n+3)$ with $\im\varrho=D$.
\end{proposition}

\begin{proof}
Assume that there exists such a PPT edge state $\varrho$.
We may suppose that $D^\perp$ is spanned by $|0\ran \otimes |0\ran$ by Schmidt decomposition.
We put $E = \im \varrho^\Gamma$ and write
$E^\perp = {\rm span} \{ v_0, \dots,  v_{(2n-3)-1} \}$. We denote by $\tilde v_k$ the compression of $v_k$ onto $\mathbb C^{n-1} \otimes \mathbb C^n$.
Then we have
$$
\begin{aligned}
\dim {\rm span}\{ \tilde v_0, \dots,  \tilde v_{(2n-3)-1} \}^\perp
& \ge (n-1)n -(2n-3)
&> ((n-1)-1)(n-1).
\end{aligned}
$$
It follows from \cite{part} that ${\rm span}\{ \tilde v_0, \dots,  \tilde v_{(2n-3)-1} \}^\perp$ contains a product vector
$$
|\bar \xi' \ran \otimes |\eta\ran = (\bar x_1, \dots, \bar x_{n-1})^\ttt \otimes (y_0, \dots,  y_{n-1})^\ttt.
$$
We put $|\xi\ran = (0,x_1,\dots,  x_{n-1})^\ttt$.
Then, we have
$$
\lan 0|\lan 0 |\xi\ran |\eta\ran =0, \qquad \text{and} \qquad \lan v_k|\bar\xi \ran |\eta\ran = \lan \tilde v_k|\bar\xi' \ran |\eta\ran = 0, \qquad 0 \le k \le (2n-3)-1,
$$
which shows that $\varrho$ is not an edge state.
\end{proof}

It would be interesting to know if there exists a PPT entangled edge state whose range is a maximal subspace with the orthogonal complement
spanned by a vector of Schmidt rank one.
In the previous work \cite{han_kye_2025_a} by the authors, we showed that the distribution of Schmidt number witnesses outside
of faces ${\mathcal F}_E$ is determined by the \lq entangledness\rq\ properties of the orthogonal complement $E^\perp$.
In particular, we have seen that
$|\zeta\ran\in\mathbb C^n\ot\mathbb C^n$ is of Schmidt rank one if and only if there exists no entanglement
witness outside of the face ${\mathcal F}_{|\zeta\ran^\perp}$. We note that this is the case if and only if the maximal face
${\mathcal F}_{|\zeta\ran^\perp}$ lies on the boundary of convex set $\blockpos$ of all blockpositive matrices with trace one.
See \cite{{kye_comp_tensor},{kye_lec_note}} for exposition on the notions of entanglement witnesses and blockpositivity.
Proposition \ref{rankoneedge} suggests that this property has a close relation with the existence of
PPT edge states of certain types in the interior of ${\mathcal F}_{|\zeta\ran^\perp}$.

We finish this paper to consider general $n\ot n$ cases with $n\ge 4$.
It is not so difficult to follow the construction (\ref{varrho}) of PPT states,
to get PPT states with rank $n^2-1$ for arbitrary $n\ge 4$. In this case,
the rank of $\varrho^\Gamma$ is given by
$$
n+\left(\begin{matrix}n\\ 2\end{matrix}\right)=\frac {n(n+1)}2,
$$
which is strictly less than the largest possible rank $n^2-2n+3$ for $n\ge 4$.
PPT states of type $(n^2-1,n^2-2n+3)$ have been constructed in \cite{choi_kiem_kye_2020} for arbitrary $n\ge 3$,
and some of them were shown to be edge states for $3\le n\le 1000$.
As for PPT states of other types, there are numerical results \cite{LMS_2010a} for the $4\ot 4$ case.
It would be interesting to exhibit analytic examples of $4\ot 4$ PPT edge states of various types, following our methods in the previous section.


\end{document}